\documentclass[conference]{IEEEtran}
\IEEEoverridecommandlockouts

\usepackage{cite}
\usepackage{amsmath,amssymb,amsfonts}
\usepackage{algorithmic}
\usepackage{graphicx}
\usepackage{textcomp}
\usepackage{xcolor}

\usepackage{xspace}
\usepackage[ruled,vlined]{algorithm2e}
\usepackage{booktabs}
\usepackage{pifont}
\usepackage{fancyhdr}
\usepackage{color}
\usepackage{soul, framed}
\usepackage{makecell}
\usepackage{colortbl}
\usepackage{xcolor,colortbl}
\usepackage{multirow}
\usepackage{bm,dsfont}
\usepackage{subfigure}
\usepackage{balance}
\usepackage{stmaryrd}
\usepackage{autobreak}
\usepackage{enumitem}
\usepackage{url}
\usepackage{amsthm}

\def\BibTeX{{\rm B\kern-.05em{\sc i\kern-.025em b}\kern-.08em
    T\kern-.1667em\lower.7ex\hbox{E}\kern-.125emX}}

\newtheorem{example}{\textbf{Example}}

\newtheorem{lemma}{\textbf{Lemma}}

\newtheorem{definition}{Definition}

\newlength{\figsize} 

\renewcommand{\footnotesize}{\small}

\newcommand{\abcore}{$(\alpha,\beta)$-core\xspace}

\newcommand{\abcores}{$(\alpha,\beta)$-cores\xspace}

\newcommand{\online}{\texttt{Online}\xspace} 
\newcommand{\baseline}{\texttt{GCC}\xspace}
\newcommand{\adv}{\texttt{GCC+}\xspace}
\newcommand{\qf}{\texttt{GFQ}\xspace}
\newcommand{\boge}{\texttt{BiCore}\xspace} 
\newcommand{\fang}{\texttt{BIR}\xspace}

\newcommand{\parindex}{\texttt{PAR}\xspace}

\begin{document}

\title{Efficient \abcore Computation and On-the-fly Query at Billion Scale with GPUs}
\makeatletter
\newcommand{\linebreakand}{%
  \end{@IEEEauthorhalign}
  \hfill\mbox{}\par
  \mbox{}\hfill\begin{@IEEEauthorhalign}
}
\makeatother
\author{{Qingshuai Feng$^{1}$, Shunyang Li$^{2}$, Kai Wang$^{3}$, Xuemin Lin$^{3}$, Kongzhang Hao$^{4}$, Long Yuan$^{5}$}

\vspace{1.6mm}\\
\fontsize{10}{10}
\selectfont\itshape
$^1$Great Bay University; $^2$Alibaba Group; $^3$Shanghai Jiao Tong University; \\$^4$University of New South Wales; $^5$Wuhan University of Technology\\
\fontsize{9}{9} \selectfont\ttfamily\upshape
qsfeng@gbu.edu.cn;
lishunyang.lsy@alibaba-inc.com;
cskaelwang@gmail.com;\\
xuemin.lin@sjtu.edu.cn;
haokongzhang@gmail.com;
longyuan@whut.edu.cn\\}


\maketitle

\begin{abstract}
In bipartite graphs, \abcore is a widely used model for cohesive subgraph mining. Specifically, an \abcore is a maximal subgraph in which each vertex in the upper layer has degree at least $\alpha$, and each vertex in the lower layer has degree at least $\beta$. 
The state-of-the-art CPU-based solutions incur extensive costs to construct an index structure for all $\alpha$ and $\beta$ combinations, leading to scalability challenges on large bipartite graphs. 
Moreover, on-the-fly queries, which aim to determine whether an edge update belongs to a target \abcore, are essential for real-time applications such as fraud monitoring and recommendation systems. However, existing index-based methods struggle to support such queries at scale due to their high maintenance overhead.
In this paper, we investigate how to leverage GPU architectures to enable efficient \abcore computation and support on-the-fly queries. While GPUs are widely used to accelerate graph processing, their limited memory capacity makes it impractical to store large index structures. 
To address this issue, we propose \texttt{GCC}, an index-free GPU-based peeling algorithm that accelerates \abcore computation via warp-centric processing.
To further improve efficiency, we develop \texttt{GCC+}, which leverages the nested property of \abcore with a core-based early pruning strategy. 
For handling on-the-fly queries, we propose \texttt{GFQ}, a connectivity-aware algorithm that significantly narrows the computation scope by leveraging connected component information, thereby avoiding full-graph peeling.
Extensive experiments on 11 datasets demonstrate that our proposed techniques outperform existing CPU-based solutions in terms of both space and time efficiency.

\end{abstract}


\section{Introduction}
\label{sct:introduction}

Bipartite graphs are widely used to model interactions between two distinct types of entities, such as author--paper networks~\cite{konect:DBLP}, user--item networks~\cite{beutel2013copycatch}, and gene co-expression networks~\cite{kaytoue2011mining}.
Cohesive subgraph discovery is a fundamental problem in graph analytics, aiming to identify densely connected subgraphs with strong internal connectivity.
It has been widely used in applications such as community detection~\cite{BohuaYang2019uncertain, yuan2016diversified, wang2023}, social recommendation~\cite{DBLP:conf/icde/LuoLZGL22, liu2019}, and anomaly detection~\cite{yuan2016diversified, DBLP:conf/icdm/ShinEF16, he2023scaling}.
Among various cohesive models tailored for bipartite graphs, the \abcore has been widely studied~\cite{liu2019,luo2023efficient,wang2022discovering,ding2017efficient,ahmed2007visualisation,cerinvsek2015generalized}.
It is defined as a maximal subgraph in which each vertex in the upper layer has degree at least $\alpha$, and each vertex in the lower layer has degree at least $\beta$.
For instance, as shown in Figure~\ref{fig:example}, the $(2,3)$-core is the subgraph enclosed by gray dots, which contains vertices $u_0$--$u_3$ and $v_0$--$v_2$.

\begin{figure}[htb]
\centering
\includegraphics[trim=0 0 0 0,width=0.35\textwidth]{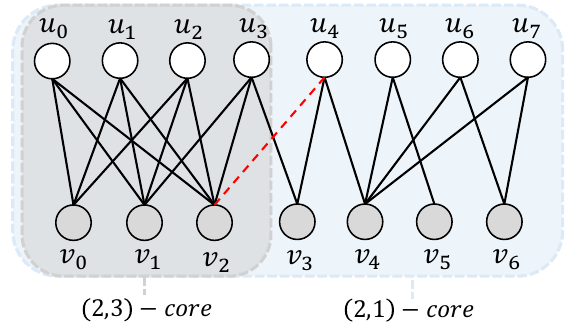}
\vspace*{-2mm}
\caption{ A bipartite graph  $G$.}
\label{fig:example}
\vspace*{-6mm}
\end{figure}

In many real-world scenarios, efficient computation of \abcores is essential to support large-scale analysis and real-time applications~\cite{wang2022discovering, liu2019, li2024querying}.
The peeling-based method~\cite{ding2017efficient} computes an \abcore by iteratively removing vertices that violate the degree constraints, but each query is processed from scratch and can be costly on large bipartite graphs.
Existing solutions~\cite{liu2019, luo2023efficient} therefore adopt index-based approaches to enhance the efficiency of \abcore computation and query processing.
However, these approaches incur substantial computational and memory overhead for constructing index structures across all $(\alpha,\beta)$ combinations, posing significant scalability challenges on large-scale bipartite graphs.
For instance, on our evaluated \texttt{PM} dataset with $4 \times 10^8$ edges, existing index-based solutions~\cite{liu2020efficient,luo2023efficient} consume over 250 GB of memory during index construction with 32 threads, making them impractical for billion-scale graph analytics.

Moreover, bipartite graphs in the wild evolve edge by edge, as observed in online social networks, web graphs, and user-product networks~\cite{li2024querying, liu2020efficient, luo2023efficient}.
Consider a user-product platform: when an order request creates a new edge $e=(u,v)$ between user $u$ and product $v$, the system needs to decide immediately whether to approve or block the payment.
A proven way to detect coordinated fraud is to check whether the user, and hence the new edge, is involved in a suspiciously dense subgraph, such as an \abcore, before the transaction is finalized~\cite{DBLP:conf/kdd/YuWW00ZLZ23}.
We formalize this need with an on-the-fly \abcore query: given a newly arrived edge update $e=(u,v)$ in a bipartite graph $G$, determine immediately after the edge insertion or deletion whether $e$ belongs to an \abcore of the updated graph.
Although some works~\cite{liu2020efficient,luo2023efficient} have introduced index-based algorithms to support dynamic edge updates, their maintenance procedures can still degenerate to the same complexity as rebuilding the index from scratch in the worst case.
Such high computational overhead makes it difficult to ensure real-time responsiveness in evolving graphs with frequent updates, rendering index-based solutions impractical for latency-sensitive applications.

Modern Graphics Processing Units (GPUs), explicitly architected for massive parallelism, have been widely adopted to accelerate graph applications such as butterfly counting~\cite{xu2022efficient}, bitruss decomposition~\cite{li2025efficient} and core decomposition~\cite{ahmad2023accelerating}.
Unlike traditional CPUs, GPUs comprise thousands of cores that execute instructions concurrently, providing substantial throughput for highly parallel graph workloads~\cite{li2025efficient}.
However, adapting index-based \abcore approaches to GPUs is infeasible due to limited GPU memory resources~\cite{xu2022efficient, ahmad2023accelerating, wang2024parallelization}.
In this paper, we aim to develop efficient GPU-based algorithms for \abcore computation and query processing, while facing the following challenges.
 
\noindent
\underline{\textit{1) Memory overhead.}}
Fast \abcore query processing often relies on index structures, which can consume substantial memory and exceed GPU capacity on large bipartite graphs.

\noindent
\underline{\textit{2) Workload imbalance.}}
Without large indexes, \abcore computation relies on iterative peeling with frequent neighbor updates. Skewed degrees in real-world bipartite graphs lead to irregular memory access and imbalanced GPU workloads.

\noindent
\underline{\textit{3) Large affected scope for on-the-fly queries.}}
For edge updates, maintaining indexes can be costly, while online peeling may involve a large affected scope for each update, making low-latency on-the-fly query processing difficult on GPUs.

\noindent
{\bf Our solutions.}
To address the memory and workload challenges in \abcore computation, we first propose a \underline{G}PU-friendly two-phase $(\alpha,\beta)$-\underline{C}ore \underline{C}omputation algorithm (\baseline), which follows the peeling paradigm~\cite{ding2017efficient} without constructing large global indexes.
In the scan phase, \baseline uses a flat ID-based filtering strategy to check all vertices in parallel, avoiding separate processing of the two bipartite layers.
In the peel phase, it adopts a warp-centric execution strategy, where each warp collaboratively processes a vertex and its neighbors, reducing thread divergence and improving memory access efficiency.
To further reduce redundant computation, we develop \adv, which performs lightweight preprocessing to compute vertex core numbers and uses them to prune vertices whose membership can be safely determined before peeling.

To address the large affected scope in on-the-fly queries, we propose a \underline{G}PU-based on-the-\underline{F}ly \abcore \underline{Q}uery algorithm (\qf).
Instead of maintaining expensive global indexes or recomputing over the entire graph, \qf uses connected component information to localize the computation after an edge update.
It then performs GPU-based peeling only within the affected component, thereby reducing both update processing cost and query latency.

\noindent
{\bf Contributions.} Our principal contributions are as follows.

\noindent
{\em $\bullet$}
We are the first to study \abcore computation and on-the-fly \abcore query processing on GPUs.

\noindent
{\em $\bullet$}
We propose \baseline, a GPU-based peeling algorithm for \abcore computation.
It combines flat ID-based parallel scanning with warp-centric peeling to improve GPU utilization under skewed bipartite graph structures.

\noindent
{\em $\bullet$}
We develop \adv, an optimized version of \baseline that uses lightweight core-number-based pruning to reduce redundant peeling operations.

\noindent
{\em $\bullet$}
We propose \qf, a GPU-based on-the-fly query algorithm that localizes computation to affected connected components after edge updates, avoiding full-graph recomputation.

\noindent
{\em $\bullet$}
We conduct extensive experiments on 11 bipartite graphs to evaluate \baseline, \adv, and \qf.
Compared with online methods that perform no preprocessing, \adv achieves at least one order of magnitude speedup in query time.
Compared with leading index-maintenance approaches, \qf achieves up to 175$\times$ speedup for on-the-fly queries.

\noindent
{\bf Roadmap.}
The rest of this paper is organized as follows.
Section~\ref{sct:related} reviews related work.
Section~\ref{sct:preliminaries} introduces preliminaries.
Sections~\ref{sct:baseline} and~\ref{sct:adv1} present the GPU-based \abcore computation and on-the-fly query algorithms, respectively.
Section~\ref{sct:experiment} reports the experimental results, and Section~\ref{sct:conclusion} concludes the paper.
\section{Related Work}
\label{sct:related}

To the best of our knowledge, this paper is the first to study \abcore computation and on-the-fly query processing by leveraging GPU architectures.
Therefore, we review two closely related topics: CPU-based cohesive subgraph computation on bipartite graphs and GPU-accelerated cohesive subgraph computation.

\noindent
{\em \bf{CPU-based cohesive subgraph computation on bipartite graphs.}}
A variety of cohesive subgraph models have been proposed for bipartite graphs, including the \abcore model~\cite{ahmed2007visualisation, cerinvsek2015generalized, ding2017efficient, liu2020efficient, wang2020efficient, he2022efficient, wang2022discovering, zhao2022efficient}, the $k$-bitruss model~\cite{zou2016bitruss, sariyuce2018peeling, wang2020efficient}, and biclique-based models~\cite{mukherjee2016enumerating, das2018incremental, ma2021efficient, zhang2014finding, chen2021efficient}.
For the \abcore model, the peeling-based method~\cite{ding2017efficient} iteratively removes vertices that violate the degree constraints.
The index-based method~\cite{liu2020efficient} further exploits the nested property of \abcores to support fast query answering.
Subsequent studies extend the \abcore model to different settings, such as weighted bipartite graphs~\cite{wang2020efficient}, hierarchical structure discovery~\cite{wang2022discovering}, and uncertain bipartite graphs~\cite{zhao2022efficient}.
Beyond \abcores, the $k$-bitruss model captures edge cohesion based on butterflies, where each edge is contained in at least $k$ butterflies.
Existing CPU-based algorithms study online peeling, index-based computation, and distributed processing for bitruss decomposition~\cite{zou2016bitruss, sariyuce2018peeling, wang2020efficient}.
Biclique-based models have also been extensively studied, including maximal biclique enumeration~\cite{mukherjee2016enumerating, das2018incremental, zhang2014finding} and maximum balanced biclique search~\cite{chen2021efficient}.
These studies provide important cohesive models and CPU-based computation techniques for bipartite graphs.
In contrast, our work focuses on GPU-based, index-free \abcore computation and on-the-fly query processing for large-scale bipartite graphs.

\noindent
{\em \bf{GPU-accelerated cohesive subgraph computation.}}
GPUs have been widely used to accelerate cohesive subgraph computation due to their massive parallelism and high memory bandwidth.
On unipartite graphs, existing studies accelerate $k$-core decomposition using peeling-based methods, vector primitives, and warp-centric strategies~\cite{alok_2018_kcore,amir_2020_kcore,ahmad2023accelerating}.
However, the $k$-core decomposition approaches cannot be directly applied to \abcore computation, because \abcores involve two vertex layers with asymmetric degree constraints and highly skewed degree distributions.
For $k$-truss decomposition, GPU-based methods have been developed to accelerate triangle-based support computation and peeling~\cite{date2017collaborative,mailthody2018collaborative,almasri2019update,che2020accelerating}.
For clique enumeration, GPU acceleration has been explored through parallel search and memory-access optimizations~\cite{wei2021accelerating}.

On bipartite graphs, GPU-based studies mainly focus on butterfly-related primitives and biclique enumeration.
For example, GPU-based butterfly counting accelerates the enumeration of $(2,2)$-bicliques by improving load balancing and memory access efficiency~\cite{xu2022efficient}.
Recent work on GPU-based bitruss decomposition leverages GPU parallelism to accelerate butterfly-support computation and peeling in bipartite graphs~\cite{li2025efficient}.
GPU-based maximal biclique enumeration further addresses the large memory requirement, thread divergence, and workload imbalance of enumerating maximal bicliques on GPUs~\cite{pan2023efficient}.
Although these GPU-based methods demonstrate the effectiveness of GPU acceleration for cohesive subgraph computation, none of them studies GPU-based \abcore computation or on-the-fly \abcore query processing.
These tasks require handling asymmetric degree constraints and avoiding large index structures under limited GPU memory.

\section{Preliminaries}
\label{sct:preliminaries}

Our problem is defined over an undirected and unweighted bipartite graph $G=(V(G), E(G))$, where the vertex set is defined as $V(G) = U(G) \cup L(G)$, with $U(G)$ and $L(G)$ denoting the disjoint sets of upper- and lower-layer vertices, respectively.
The edge set is denoted as $E(G) \subseteq U(G) \times L(G)$.
For a vertex $u \in V(G)$, its neighbor set is defined as $N(u) = \{v \in V(G) \mid (u, v) \in E(G)\}$, and its degree is $deg(u, G) = |N(u)|$.
We use $n$ and $m$ to denote the numbers of vertices and edges in $G$, respectively.
To facilitate efficient graph processing, we store the graph in the Compressed Sparse Row (CSR) format~\cite{greathouse2014efficient}, a compact and memory-efficient representation that is suited for modern memory hierarchies.

\begin{definition} \textbf{\abcore.}
\label{def:abcore}
Given a bipartite graph $G$ and two parameters $\alpha$ and $\beta$, a subgraph $\mathcal{G} \subseteq G$ is the \abcore if (1) $deg(u, \mathcal{G}) \geq \alpha$ for every vertex $u \in U(\mathcal{G})$, and $deg(v, \mathcal{G}) \geq \beta$ for every vertex $v \in L(\mathcal{G})$; (2) $\mathcal{G}$ is maximal, i.e., no supergraph $\mathcal{G^{'}}$ $\supset$ $\mathcal{G}$ qualifies as an \abcore.
\end{definition}

In this paper, we study the following two problems regarding the \abcore model.

\noindent
\textbf{Problem Statement (\abcore Computation). }
Given a bipartite graph $G$, degree constraints $\alpha$ and $\beta$, we aim to retrieve the vertex set of the \abcore in $G$.

\noindent
\textbf{Problem Statement (On-the-fly \abcore query). }
Given a bipartite graph $G$, degree constraints $\alpha$ and $\beta$, and an edge update $e = (u, v)$ (insertion or deletion), we aim to determine whether $e$ belongs to the \abcore of the updated graph. 

\begin{example} 
Figure~\ref{fig:example} shows a bipartite graph $G$, where $n=15$ and $m=20$. The subgraph enclosed by the blue dots forms the $(2,1)$-core, which contains all vertices in the bipartite graph. The $(2,3)$-core is the subgraph enclosed by gray dots, which contains $\{u_0,u_1,u_2,u_3, v_0, v_1, v_2\}$. Consider the insertion of a new edge $e=(u_4, v_2)$, shown as the red dashed line in Figure~\ref{fig:example}. An \emph{on-the-fly} query can be performed to determine whether this edge belongs to the updated $(1,5)$-core. After the insertion, vertex $v_2$ reaches degree 5, satisfying the constraint $\beta=5$, and $u_4$ trivially satisfies $\alpha=1$. Therefore, the edge $e$ belongs to the $(1,5)$-core. 
\end{example}

\subsection{CPU-based \abcore Computation}

Since no prior work has studied \abcore computation on GPUs, we first review CPU-based solutions, including the peeling-based method~\cite{ding2017efficient} and index-based methods~\cite{liu2019,luo2023efficient}, which form the main baselines in our experiments.

The peeling-based method~\cite{ding2017efficient}, denoted by \online, computes an \abcore directly from the input graph without preprocessing.
Given a bipartite graph $G$ and degree constraints $\alpha$ and $\beta$, \online iteratively removes all upper-layer vertices whose degrees are smaller than $\alpha$ and all lower-layer vertices whose degrees are smaller than $\beta$, together with their incident edges.
This process continues until every remaining vertex satisfies the corresponding degree constraint, and the remaining vertices form the \abcore.
Although \online requires no preprocessing and no extra index storage, each query must perform the peeling process from scratch, which can be expensive on large bipartite graphs.

To further improve the efficiency of \abcore computation, \boge~\cite{liu2019} introduces nested computing and builds a three-level index that organizes vertices according to their valid $(\alpha,\beta)$ ranges, thereby supporting fast arbitrary \abcore queries.
\cite{liu2020efficient} further studies update algorithms for maintaining this index under dynamic graph changes.
Another index-based method, \fang~\cite{luo2023efficient}, introduces \textit{bi-core numbers} to summarize the range of \abcores that each vertex belongs to and provides corresponding maintenance algorithms for edge insertions and deletions.
Although these index-based methods can support efficient \abcore computation and querying, they incur significant memory overhead and computational costs.
Specifically, they require at least $O(m)$ space and $O(\delta \cdot m)$ time, where $\delta$ is the maximum value such that the $(\delta,\delta)$-core is non-empty.
In dynamic settings, index maintenance algorithms can be as expensive as rebuilding the entire index in the worst case, further limiting their scalability~\cite{liu2020efficient,luo2023efficient}.
In addition, despite the improved update efficiency of \fang, it still struggles to provide low-latency support for arbitrary queries under frequent updates, since query processing may still require scanning all vertices.

\section{GPU-based \abcore Computation with Nested-Aware Pruning}
\label{sct:baseline}

While index-based approaches such as \boge~\cite{liu2019} and \fang~\cite{luo2023efficient} can efficiently support arbitrary \abcore queries, they incur substantial memory overhead, limiting their direct applicability on GPUs.
For instance, the index construction processes of both \boge and \fang consume over 250 GB of memory when using 32 threads, which far exceeds the 80 GB memory capacity of modern A100 GPUs.
In contrast, the online peeling-based approach~\cite{ding2017efficient} naturally avoids index construction.
However, during iterative degree updates, multi-threaded CPU implementations of \online require frequent synchronization to ensure consistency, leading to substantial coordination overhead.
Moreover, \online alternates layer-wise scans over the two partitions to identify removable vertices, which further limits parallelism.

To support efficient \abcore computation on GPUs, we propose a \underline{G}PU-friendly two-phase $(\alpha,\beta)$-\underline{C}ore \underline{C}omputation algorithm (\baseline), which follows the peeling paradigm.
In each iteration, \baseline processes vertices in parallel and updates the degrees of their neighbors when vertices are removed.
To handle workload imbalance caused by skewed vertex degrees, \baseline adopts a warp-centric execution strategy, where each GPU warp collaboratively processes a vertex and its neighborhood.
This design improves parallel throughput without relying on costly global index structures.
To further improve overall efficiency, we develop an optimized algorithm, \adv, which integrates lightweight preprocessing based on vertex core numbers for early pruning.

\subsection{The \baseline algorithm}

We now introduce the \baseline algorithm in detail.
Unlike the CPU-based peeling approach~\cite{ding2017efficient}, which alternates between scanning upper and lower layers to identify vertices that violate the degree constraints, \baseline scans the entire graph in a unified manner to enable efficient parallel execution.
Specifically, each vertex $u \in V(G)$ is assigned a unique continuous ID, enabling GPU threads to determine whether a vertex belongs to the upper or lower layer via simple range-based checks.
CPU-based peeling typically uses conditional checks to determine layers, which can cause branch mispredictions and hurt cache performance~\cite{li2024querying, liu2019, li2025efficient}.
In contrast, our GPU method assigns IDs to reduce these branches and improve efficiency.
The input graph $G$ is stored in the Compressed Sparse Row (CSR) format~\cite{greathouse2014efficient}, which supports efficient neighbor access during peeling.
Algorithm~\ref{algorithm:baseline} presents the detailed procedure of the scan and peel phases.

In the scan phase (Lines 2--6), all vertices are processed in parallel to check whether they violate their corresponding $(\alpha,\beta)$ constraints.
Vertices that violate the constraints are inserted into the global candidate set $\mathcal{C}$ for subsequent peeling.
To ensure concurrency safety,  we employ atomic operations when updating $\mathcal{C}$, thereby reducing race conditions while maintaining high throughput.

In the peeling phase (Lines 7--15), \baseline processes the candidate set using a warp-centric strategy.
Each warp is assigned a vertex $u \in \mathcal{C}$ and collaboratively traverses its neighbors $N(u)$.
This design aligns with the CSR layout, improving memory coalescing and reducing thread divergence.
For each neighbor $v \in N(u)$, the warp performs an atomic decrement on $\deg(v,\mathcal{G})$ to reflect the removal of $u$ (Line 11).
If the updated degree of $v$ falls below its corresponding threshold, i.e., $\alpha$ for upper-layer vertices or $\beta$ for lower-layer vertices, $v$ is identified as a new candidate and inserted into a warp-local queue (Lines 12-13).
To reduce contention on the global candidate set, newly invalidated vertices are first buffered in warp-local queues and then flushed to $\mathcal{C}$ (Line 14).
Block-level synchronization is enforced between iterations to ensure that all degree updates are visible before the next round.

The algorithm terminates when $\mathcal{C}$ becomes empty and returns the remaining vertex set as the final \abcore.
Since \baseline avoids global index construction and processes each edge a constant number of times, its space and time complexity are bounded by $O(m)$.

\begin{algorithm}[t]
  \small
  \DontPrintSemicolon
  \caption{\baseline}
  \label{algorithm:baseline}
  \LinesNumbered
  \KwIn{$G=(U,L,E)$, degree thresholds $\alpha,\beta$}
  \KwOut{Vertex set of the $(\alpha,\beta)$-core of $G$}

  $\mathcal{G} \gets G$; \quad $\mathcal{C} \gets \emptyset$ \\

  \ForEach{$v \in V(\mathcal{G})$ in parallel} {
    \If{$v \in U(\mathcal{G}) \land \deg(v,\mathcal{G})<\alpha$}{
      enqueue $v$ into $\mathcal{C}$
    }
    \ElseIf{$v \in L(\mathcal{G}) \land \deg(v,\mathcal{G})<\beta$}{
      enqueue $v$ into $\mathcal{C}$
    }
  }

  \While{$\mathcal{C} \neq \emptyset$} {
    \ForEach{$\text{warp} \in \text{launched warps}$ in parallel} {
      $u \gets \mathcal{C}(\text{warp}_{\text{id}})$  \\

      \ForEach{$v \in N(u)$ in parallel} {
        atomic\_decrement $\deg(v,\mathcal{G})$  \\
        \If{($v \in U(\mathcal{G}) \land \deg(v,\mathcal{G})<\alpha$) 
           \textbf{or} ($v \in L(\mathcal{G}) \land \deg(v,\mathcal{G})<\beta$)}{
          enqueue $v$ into warp-local queue
        }
      }

      flush warp-local queue to $\mathcal{C}$
    }

    block\_synchronize
  }

  \Return{$V(\mathcal{G})$}
\end{algorithm}

\begin{figure}[t]
\vspace{-3mm}
\centering
\includegraphics[trim=0 0 0 0,width=0.48\textwidth]{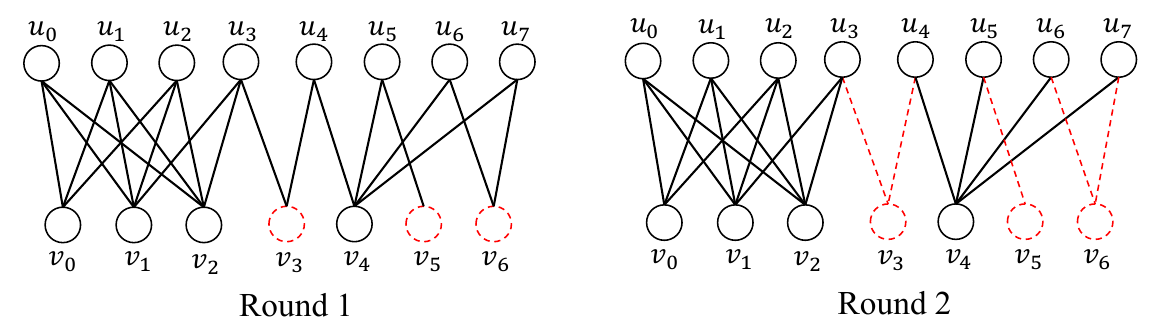}
\vspace*{-2mm}
\caption{Illustrating the $(1,3)$-core computation using \baseline.}
\label{fig:base_exp}
\vspace*{-1em}
\end{figure}

\begin{example}
Figure~\ref{fig:base_exp} illustrates the $(1,3)$-core computation process using \baseline.
In the first round, vertices $v_3$, $v_5$, and $v_6$ violate the degree threshold $\beta=3$ and are added to the candidate set for removal, as highlighted by red dashed circles.
After removing these vertices, \baseline updates the degrees of their neighbors.
In the second round, no further vertices violate the constraints, and the process terminates.
The final $(1,3)$-core consists of ${u_0,u_1,u_2,u_3,u_4,u_5,u_6,u_7,v_0,v_1,v_2,v_4}$.
\end{example}

\subsection{The \adv algorithm}
Although the GPU-based peeling algorithm \baseline substantially accelerates \abcore computation, it may still perform redundant work for queries where $\alpha$ and $\beta$ are close.
In this case, many vertices can be clearly determined before peeling, but \baseline still performs full-scale neighbor updates over them.
To reduce such unnecessary processing, we introduce an early pruning mechanism that reduces the number of active vertices involved in peeling.

Inspired by the classic $k$-core definition~\cite{khaouid2015k}, we use vertex core numbers to derive safe pruning conditions for \abcore computation.
Intuitively, the core number of a vertex indicates the largest symmetric $(\delta,\delta)$-core containing that vertex.
Thus, it can provide useful bounds for an asymmetric $(\alpha,\beta)$ query.
Vertices with insufficient core numbers can be safely excluded before peeling, while vertices with sufficiently large core numbers can be directly retained.
This lightweight preprocessing step can be efficiently parallelized on GPUs and helps reduce the search space, memory traffic, and redundant degree updates.
To utilize this strategy, we first define the notion of core numbers in bipartite graphs.

\begin{definition} \textbf{Core Number~\cite{batagelj2003m, khaouid2015k}.}
\label{def:corenumber}
Given a bipartite graph $G$, the \emph{core number} of a vertex $u \in V(G)$ is the largest integer $\delta$ such that $u$ belongs to the $(\delta,\delta)$-core of $G$. Formally,
\[
\text{core}(u) = \max \left\{ \delta \in \mathbb{N} \mid u \in \mathcal{G}, \text{ where } \mathcal{G} \text{ is a } (\delta,\delta)\text{-core of } G \right\}
\]
\end{definition}

According to the definition of the core number in bipartite graphs (i.e., Definition~\ref{def:corenumber}), we can derive a crucial pruning principle that enables early elimination or retention of vertices before the peeling process. 
This observation leads to the following lemma:

\begin{lemma}
\label{lemma:corefilter}
Given a bipartite graph $G$, let $\text{core}(u)$ denote the core number of vertex $u \in V(G)$.
When computing the $(\alpha,\beta)$-core, the following pruning conditions hold:
\begin{itemize}
\item[(1)] If $\text{core}(u) < \min(\alpha,\beta)$, then $u \notin \mathcal{G}$, where $\mathcal{G}$ is the $(\alpha,\beta)$-core of $G$.
\item[(2)] If $\text{core}(u) \geq \max(\alpha,\beta)$, then $u$ is guaranteed to be included in $\mathcal{G}$.
\end{itemize}
\end{lemma}

\begin{proof}
Let $\delta_l=\min(\alpha,\beta)$ and $\delta_q=\max(\alpha,\beta)$.

(1) If $\text{core}(u)<\delta_l$, then by Definition~\ref{def:corenumber}, vertex $u$ does not belong to any $(\delta_l,\delta_l)$-core.
Since any $(\alpha,\beta)$-core must satisfy degree thresholds at least $\delta_l$ on both layers, $u$ cannot belong to the $(\alpha,\beta)$-core.
Thus, Lemma~\ref{lemma:corefilter}(1) holds.

(2) If $\text{core}(u)\geq\delta_q$, then $u$ belongs to the $(\delta_q,\delta_q)$-core.
Since $\delta_q\geq\alpha$ and $\delta_q\geq\beta$, this $(\delta_q,\delta_q)$-core also satisfies the $(\alpha,\beta)$ degree constraints.
By the maximality of the $(\alpha,\beta)$-core, $u$ must be included in $\mathcal{G}$.
Thus, Lemma~\ref{lemma:corefilter}(2) holds.
\end{proof}

Based on Lemma~\ref{lemma:corefilter}, we can significantly reduce the workload of GPU-based online \abcore computation (i.e., \baseline) by preemptively pruning vertices that are guaranteed to be either excluded from or included in the result.
This enhancement leads to an optimized algorithm design with lower runtime overhead on large bipartite graphs. 
We next introduce an enhanced GPU-based algorithm, \adv, that incorporates the core-number-based pruning strategy into the \baseline framework.

\begin{algorithm}[htbp]
  \small
  \DontPrintSemicolon
  \caption{\adv}
  \label{algorithm:adv_baseline}
  \LinesNumbered
  \KwIn{$G$, $\alpha$, $\beta$, precomputed core number $\text{core}(u)$ for each vertex $u \in V(G)$}
  \KwOut{Vertices in the $(\alpha,\beta)$-core of $G$}

  $\delta_l \gets \min(\alpha,\beta)$; \quad $\delta_q \gets \max(\alpha,\beta)$ \\

  $\mathcal{G} \gets G$; \quad $\mathcal{C} \gets \emptyset$ \\

  \ForEach{$v \in V(G)$ in parallel} {
    \If{$\delta_l \le \text{core}(v) < \delta_q$}{
        \If{$v \in U(\mathcal{G}) \land \deg(v,\mathcal{G})<\alpha$}{
          enqueue $v$ into $\mathcal{C}$
        }
        \ElseIf{$v \in L(\mathcal{G}) \land \deg(v,\mathcal{G})<\beta$}{
          enqueue $v$ into $\mathcal{C}$
        }
    }
  }

  \While{$\mathcal{C} \neq \emptyset$} {
    \ForEach{$\text{warp} \in \text{launched warps}$ in parallel} {
      $u \gets \mathcal{C}(\text{warp}_{\text{id}})$  \\
      \ForEach{$v \in N(u)$ in parallel} {
        \If{$\delta_l \le \text{core}(v) < \delta_q$}{
          atomic\_decrement $\deg(v,\mathcal{G})$  \\
          \If{$(v\in U(\mathcal{G}) \text{ and } \deg(v,\mathcal{G})<\alpha)$ 
             \textbf{or} $(v\in L(\mathcal{G}) \text{ and } \deg(v,\mathcal{G})<\beta)$}{
            enqueue $v$ into warp-local queue
          }
        }
      }
      flush warp-local queue to $\mathcal{C}$
    }
    block\_synchronize
  }

  \Return{$V(\mathcal{G})$}
\end{algorithm}

Algorithm~\ref{algorithm:adv_baseline} outlines \adv, an optimized GPU-based $(\alpha,\beta)$-core computation algorithm that integrates core-number-based pruning to significantly reduce the search space. Unlike the baseline \baseline, which indiscriminately inspects all neighbors during peeling, \adv leverages precomputed core numbers (Definition~\ref{def:corenumber}) to filter vertices statically.

The algorithm begins by deriving two pruning thresholds: $\delta_l = \min(\alpha, \beta)$ and $\delta_q = \max(\alpha, \beta)$ (Line 1). According to Lemma~\ref{lemma:corefilter}, vertices with $core(u) < \delta_l$ are inherently excluded, while those with $core(u) \ge \delta_q$ satisfy the degree constraints trivially and require no dynamic verification. Consequently, the computation focuses solely on the "uncertain" vertices within the range $[\delta_l, \delta_q)$.

In the scan phase (Lines 3–8), the algorithm populates the initial candidate set $\mathcal{C}$. A thread is assigned to each vertex $v$ to check two conditions: (1) whether $core(v)$ falls within the pruning range $[\delta_l, \delta_q)$, and (2) whether it violates the degree constraints. Only vertices satisfying both conditions are enqueued into $\mathcal{C}$, effectively filtering out irrelevant updates before the peeling begins.

The peeling phase (Lines 9–18) employs a warp-centric strategy to mitigate thread divergence and memory contention. Each warp processes a vertex $u \in \mathcal{C}$ and its neighbors $N(u)$ in parallel. A critical optimization occurs at Line 13: before performing an atomic decrement on a neighbor $v$, the warp verifies if $core(v) \in [\delta_l, \delta_q)$. This ensures that updates are only propagated to vertices whose status is not yet determined, thereby avoiding redundant atomic operations on vertices that are guaranteed to remain in the core. If a neighbor $v$ violates the $(\alpha, \beta)$ constraints after the update, it is buffered in a warp-local queue and subsequently flushed to the global set $\mathcal{C}$. Global consistency is maintained via block-level synchronization (Line 18). The process terminates when $\mathcal{C}$ becomes empty, returning the exact $(\alpha,\beta)$-core. 

\begin{figure}[t]
\begin{centering}
\subfigure[\texttt{Core number of a bipartite graph $G$}]{
\begin{minipage}[b]{0.4\textwidth}
\includegraphics[trim=0 0 0 0,clip,width=1\textwidth]{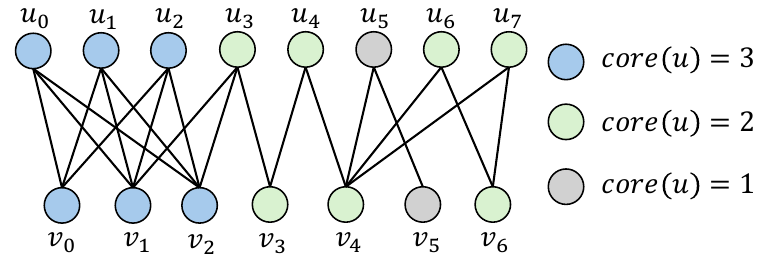} 
\label{fig:r1}
\vspace{-5mm}
\end{minipage}}
\subfigure[\texttt{Peeling process}]{
\begin{minipage}[b]{0.45\textwidth}
\includegraphics[trim=0 0 0 0,clip,width=1\textwidth]{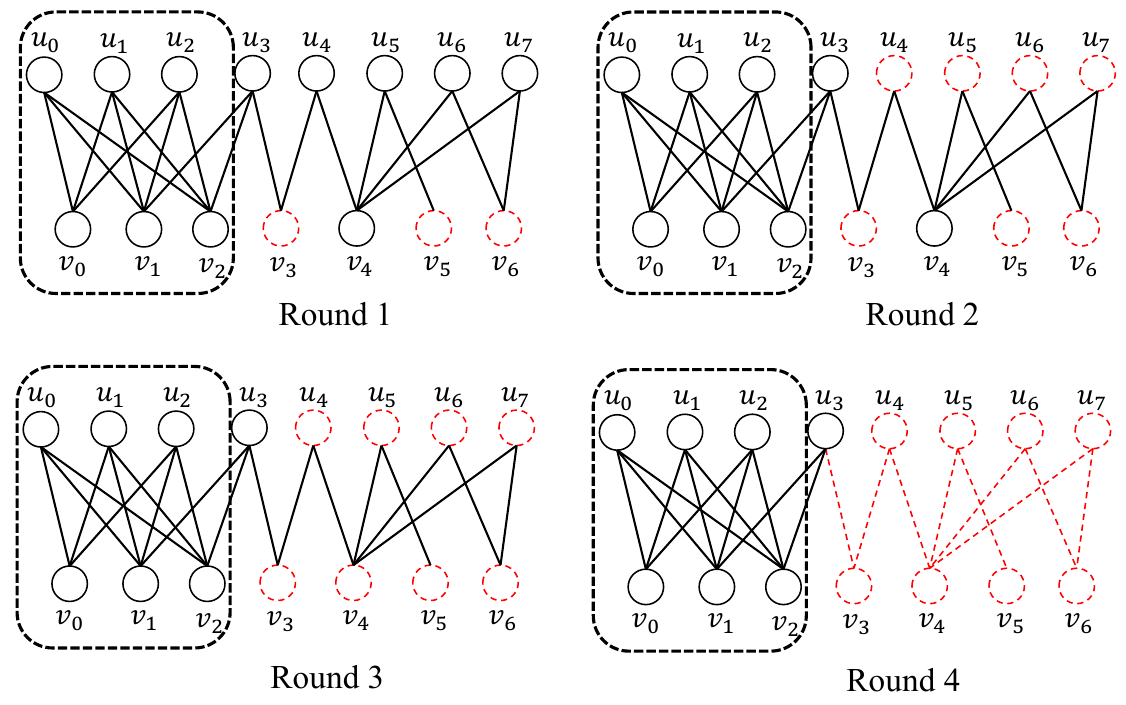} 
\label{fig:r2}
\vspace{-5mm}
\end{minipage}} 
\caption{Illustrating the $(2,3)$-core computation using \adv.}
\label{fig:exp2}
\end{centering}
\vspace*{-3mm}
\end{figure}

{\color{black}
\begin{example}

Figure~\ref{fig:exp2} demonstrates the $(2,3)$-core computation using the \adv algorithm. 
In Figure~\ref{fig:r1}, the preprocessing step computes the core number of each vertex, where blue, green, and gray nodes represent core numbers of 3, 2, and 1, respectively. 
Given the thresholds $\alpha=2$ and $\beta=3$, we determine pruning bounds $\delta_l=2$ and $\delta_q=3$.
Vertices with core number at least $\delta_q$ (i.e., those enclosed by the dashed boundary, such as $u_0$ to $u_3$ and $v_0$ to $v_3$) are guaranteed to satisfy the constraints and are excluded from the peeling. 
In the first round, vertices $v_3$, $v_5$, and $v_6$ are identified as violating the degree threshold $\beta=3$ (highlighted as red dashed circles) and are added to the candidate set for removal. 
Their removal causes degree updates to their neighbors.
In Round 2, vertices $u_4$, $u_5$, $u_6$, and $u_7$ fall below the threshold $\alpha=2$ and are peeled accordingly.
The process terminates when no further violations exist, and the final $(2,3)$-core includes vertices $\{u_0, u_1, u_2, u_3, v_0, v_1, v_2\}$.
\end{example}
}
\section{On-the-Fly Query Processing}
\label{sct:adv1}

In many real-world applications, such as user--item networks~\cite{beutel2013copycatch}, social networks~\cite{cerinvsek2015generalized}, and financial networks~\cite{luo2023efficient}, graphs are inherently dynamic, with edges and vertices frequently inserted or deleted over time.
A major challenge in such scenarios is to support real-time query processing on dynamic bipartite graphs.
For example, in a user--item bipartite graph on an e-commerce platform, when a user purchases a product, the system may need to immediately determine whether this interaction is involved in suspicious behavior, such as forming a dense structure with highly unbalanced interaction patterns between the two sides of the graph~\cite{li2024querying}.
In such cases, fast on-the-fly structural queries are essential for timely fraud detection.

To address this challenge, we propose \qf, a \underline{G}PU-based on-the-\underline{F}ly \abcore \underline{Q}uery algorithm that leverages graph connectivity to reduce the computation scope.
The key observation is that the impact of an edge update is restricted to the connected component region affected by that update.
By maintaining component information, \qf avoids applying peeling to the entire graph and focuses only on the subgraph that may influence the query result.

\noindent{\bf Preprocessing.}
Instead of maintaining expensive global indexes, \qf uses a lightweight GPU-parallel union-find procedure to assign a component identifier to each vertex.
These component identifiers allow \qf to quickly locate the affected component when an edge update arrives.

\noindent{\bf Query procedure.}
Given an edge update $e=(u,v)$, \qf first updates the component information.
If $e$ is an insertion between two different components, the two components are merged.
If $e$ is a deletion, \qf re-invokes the parallel connectivity procedure on the affected component to handle possible component splits.
Then, \qf extracts the induced subgraph $\mathcal{G}$ associated with the affected component and executes the peeling process, i.e., Algorithm~\ref{algorithm:baseline}, only on $\mathcal{G}$.
This avoids redundant computation on unrelated parts of the graph.
Finally, \qf returns whether the updated edge or its endpoints remain in the target \abcore of the affected subgraph, depending on the update type.

\noindent
{\bf Remark.}
The algorithm \qf is not limited to verifying the membership of an updated edge $e=(u,v)$.
It can also be adapted to retrieve the vertex set of the \abcore containing the updated endpoints by returning the remaining vertices after peeling the affected component.
Moreover, in the static setting, \qf can support localized community search by returning the \abcore that contains a query vertex or edge, without requiring a global index or full-graph traversal.

\noindent
{\bf Complexity analysis of \qf.}
Let $G$ be partitioned into $x$ disjoint connected components $C_1,C_2,\ldots,C_x$, where each component $C_i=(U_i,L_i,E_i)$ contains a subset of vertices and edges.
After the affected component $C_i$ is identified, the query cost of \qf is bounded by $O(|E_i|)$, since peeling is performed only within $C_i$.
In the worst case, if the graph contains a single connected component, the complexity becomes $O(m)$.
However, in practice, many updates affect only local components, and the component-aware design significantly reduces unnecessary computation.

\section{Experimental Evaluation}
\label{sct:experiment}

\subsection{Experimental setting}

We compare our methods with the following designs for \abcore computation and on-the-fly query processing.

\noindent
\textbf{CPU-based algorithms.}
The CPU-based baselines include:
1)~\online~\cite{ding2017efficient}: the classic peeling algorithm that iteratively removes vertices violating the $(\alpha,\beta)$ degree constraints.
2)~\boge~\cite{liu2019, liu2020efficient}: an index-based method that precomputes and stores \abcore indexes for fast query answering.
3)~\fang~\cite{luo2023efficient}: a compact index-based method that maintains bicore numbers for efficient arbitrary \abcore queries.
4)~\parindex~\cite{huang2023efficient}: a parallel index-based method that adopts the same three-level index structure as \boge.

Both \boge and \fang provide maintenance methods for dynamic bipartite graphs, whereas \parindex focuses on parallel index construction and does not propose a maintenance algorithm.
However, \fang runs out of memory on all datasets when processing our dynamic workloads.
Therefore, in the on-the-fly query evaluation, we report only \boge as the representative index-maintenance baseline for comparison with our proposed \qf algorithm.

\noindent
\textbf{GPU-based algorithms.}
We evaluate the following GPU-based algorithms:
5)~\baseline: our GPU-based parallel peeling algorithm for \abcore computation, as described in Algorithm~\ref{algorithm:baseline}.
6)~\adv: an optimized version of \baseline that incorporates core-number-based pruning to reduce redundant computation, as described in Algorithm~\ref{algorithm:adv_baseline}.
7)~\qf: our GPU-based on-the-fly query algorithm that restricts computation to the connected component containing the updated edge.

\begin{table}[t]
\caption{Summary of Datasets.}
\vspace*{-2mm}
\begin{center}
    \begin{tabular}{c|c|c|c|c}
        \noalign{\hrule height 1pt}
        \textbf{Dataset} & $|U(G)|$ & $|L(G)|$ & $|E(G)|$ & $\delta$ \\\hline
        RE  & 781,265   & 283,911    & 60,569,726   & 191 \\ \hline
        NY  & 299,752   & 101,636    & 69,679,427   & 314 \\ \hline
        TR  & 551,787   & 1,173,225  & 83,629,405   & 509 \\ \hline
        Dui & 833,081   & 33,778,221 & 101,798,957  & 184 \\ \hline
        LG  & 3,201,203 & 7,489,073  & 112,307,385  & 109 \\ \hline
        Dti & 4,511,972 & 33,777,768 & 137,240,382  & 180 \\ \hline
        WT  & 27,665,730& 12,756,244 & 140,613,762  & 438 \\ \hline
        YS  & 1,000,990 & 624,961    & 256,804,235  & 1,100 \\ \hline
        OG  & 2,783,196 & 8,730,857  & 327,037,487  & 467 \\ \hline
        PM  & 8,200,000 & 141,043    & 483,450,157  & 108 \\ \hline
        PL  & 220,576 & 5,000,000  & 1,000,000,000 & 215 \\ \hline
        \noalign{\hrule height 1pt}
    \end{tabular}
\end{center}
\label{table:datasets}
\vspace{-3.5em}
\end{table}

\noindent
\textbf{Datasets.}
We evaluate all algorithms on 11 bipartite graphs, including 10 real-world datasets from KONECT (\url{http://konect.cc/}) and one synthetic dataset.
The synthetic dataset \texttt{PL} follows a power-law distribution, which is widely used to mimic structural properties of real-world networks~\cite{liu2019,luo2023efficient}.
Table~\ref{table:datasets} summarizes the dataset statistics, where $|U(G)|$ and $|L(G)|$ denote the numbers of vertices in the upper and lower layers, respectively, and $|E(G)|$ denotes the number of edges.
The parameter $\delta$ is the largest value such that the $(\delta,\delta)$-core of $G$ is non-empty.

\noindent
\textbf{Experimental environment.}
All GPU-based algorithms are implemented in CUDA, while all CPU-based algorithms are implemented in C++.
By default, GPU-based algorithms are evaluated on a Linux server equipped with an NVIDIA RTX 3090 Ti GPU with 24 GB of global memory.
For comparison, CPU-based algorithms are executed on a Linux server with an Intel Xeon Silver 4314 processor and 512 GB of main memory.
An algorithm is terminated if its execution time exceeds 10 hours, and dynamic update algorithms are terminated if their running time exceeds 3600 seconds.
The abbreviation \texttt{OOM} denotes out-of-memory errors.

\subsection{\abcore Computation Performance}

\begin{table*}[ht]
\centering
\small
\caption{Query efficiency evaluation: preprocessing time and query time (in seconds).}
\vspace{-2mm}
\label{tab:efficiency_evaluation}
\resizebox{0.98\textwidth}{!}{
\setlength{\tabcolsep}{5pt}
    \begin{tabular}{l *{6}{c} *{6}{c}}
        \toprule
        \multirow{2}{*}{\textbf{Dataset}} & 
        \multicolumn{6}{c}{\textbf{Preprocessing Time (sec)}} & 
        \multicolumn{6}{c}{\textbf{Query Time (sec)}} \\
        \cmidrule(lr){2-7} \cmidrule(lr){8-13}
        & \boge & \fang & \parindex & \online & \baseline & \adv 
        & \boge & \fang & \parindex & \online & \baseline & \adv \\
        \midrule
        \textit{RE}  & 53.84 & 21.05 & 11.27 & 0.00 & 0.00 & 0.19 & 0.0010 & 0.0233 & 0.0015 & 0.2865 & 0.0100 & 0.0077 \\
        \textit{NY}  & 61.33 & 28.39 & 12.95 & 0.00 & 0.00 & 0.14 & 0.0012 & 0.0108 & 0.0020 & 0.1241 & 0.0040 & 0.0029 \\
        \textit{TR}  & 101.77 & 54.16 & 24.68 & 0.00 & 0.00 & 0.28 & 0.0005 & 0.0176 & 0.0004 & 0.3449 & 0.0072 & 0.0053 \\
        \textit{Dui} & 422.26 & 90.64 & 110.13 & 0.00 & 0.00 & 0.11 & 0.0003 & 0.0918 & 0.0005 & 1.6082 & 0.0467 & 0.0280 \\
        \textit{LG}  & 186.17 & 41.89 & 41.17 & 0.00 & 0.00 & 0.09 & 0.0015 & 0.0764 & 0.0021 & 1.4509 & 0.0334 & 0.0222 \\
        \textit{Dti} & 566.51 & 130.95 & 121.59 & 0.00 & 0.00 & 0.32 & 0.0001 & 0.0881 & 0.0002 & 2.1534 & 0.0678 & 0.0344 \\
        \textit{WT}  & 667.48 & 240.45 & 214.28 & 0.00 & 0.00 & 1.03 & 0.0001 & 0.0914 & 0.0001 & 1.6731 & 0.0946 & 0.0347 \\
        \textit{YS}  & 500.80 & 347.84 & \texttt{OOM} & 0.00 & 0.00 & 0.86 & 0.0006 & 0.0194 & \texttt{OOM} & 1.3236 & 0.0284 & 0.0141 \\
        \textit{OG}  & 621.56 & 313.56 & 197.84 & 0.00 & 0.00 & 0.53 & 0.0010 & 0.0605 & 0.0015 & 2.8658 & 0.0538 & 0.0387 \\
        \textit{PM}  & 488.88 & 175.87 & \texttt{OOM} & 0.00 & 0.00 & 0.38 & 0.0206 & 0.1306 & \texttt{OOM} & 3.3818 & 0.0520 & 0.0481 \\
        \textit{PL}  & 1525.37 & 662.90 & 158.55 & 0.00 & 0.00 & 1.00 & 0.0204 & 0.8360 & 0.0219 & 0.0229 & 0.0006 & 0.0003 \\
        \bottomrule
    \end{tabular}
}
\vspace{-1.5em}
\end{table*}

\noindent
\textbf{Preprocessing cost and query efficiency.}
Table~\ref{tab:efficiency_evaluation} reports the preprocessing cost and query time of all evaluated methods across all datasets.
For fair comparison, \boge, \parindex, and \fang build their indexes using 32 CPU threads.

As shown in Table~\ref{tab:efficiency_evaluation}, \adv incurs much lower preprocessing cost than the index-based methods \boge, \parindex, and \fang on all datasets.
For example, on the largest dataset \texttt{PL}, \adv takes only 1.00 second for preprocessing, while \boge, \parindex, and \fang require 1525.37, 158.55, and 662.90 seconds, respectively.
On \texttt{Dti} and \texttt{YS}, \adv takes only 0.32 and 0.86 seconds, respectively, whereas \boge requires 566.51 and 500.80 seconds, and \fang requires 130.95 and 347.84 seconds.
Moreover, \parindex runs out of memory on \texttt{YS} and \texttt{PM}, requiring more than 512 GB memory during preprocessing.
This is because \adv performs only lightweight preprocessing by computing the core number of each vertex $u \in V(G)$, instead of maintaining global index structures.

Table~\ref{tab:efficiency_evaluation} also reports the query time under the default setting, where $\alpha = 0.4 \times \delta$ and $\beta = 0.6 \times \alpha_{\text{offset}}$, and $\alpha_{\text{offset}}$ denotes the largest $\beta$ value such that the $(\alpha,\beta)$-core is non-empty.
Both \baseline{} and \adv{} achieve millisecond-level query latency with little preprocessing overhead.
Compared with \online, \adv achieves a 34$\times$--94$\times$ speedup, while \baseline achieves a 18$\times$--65$\times$ speedup.
For example, on \texttt{YS}, \online takes 1.3236 seconds, whereas \baseline and \adv take only 0.0284 and 0.0141 seconds, respectively.
The core-number-based pruning of \adv further improves over \baseline; on \texttt{WT}, \adv reduces the query time from 0.0946 seconds to 0.0347 seconds.
Compared with CPU-based methods, \adv consistently outperforms \fang and \online in query time, and also outperforms \boge and \parindex on several large datasets.
For example, on the largest dataset \texttt{PL}, \adv answers a query in 0.0003 seconds, achieving 68$\times$ and 2787$\times$ speedups over \boge and \fang, respectively.
Although \boge and \parindex can answer some queries faster on several datasets, they require substantial preprocessing cost.
Overall, \adv provides a favorable trade-off by avoiding expensive index construction while still supporting low-latency \abcore computation.

\begin{figure}[t]
    \centering
    \includegraphics[width=\linewidth]{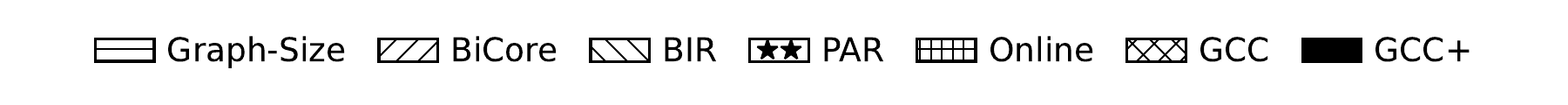}

    \includegraphics[width=0.99\linewidth]{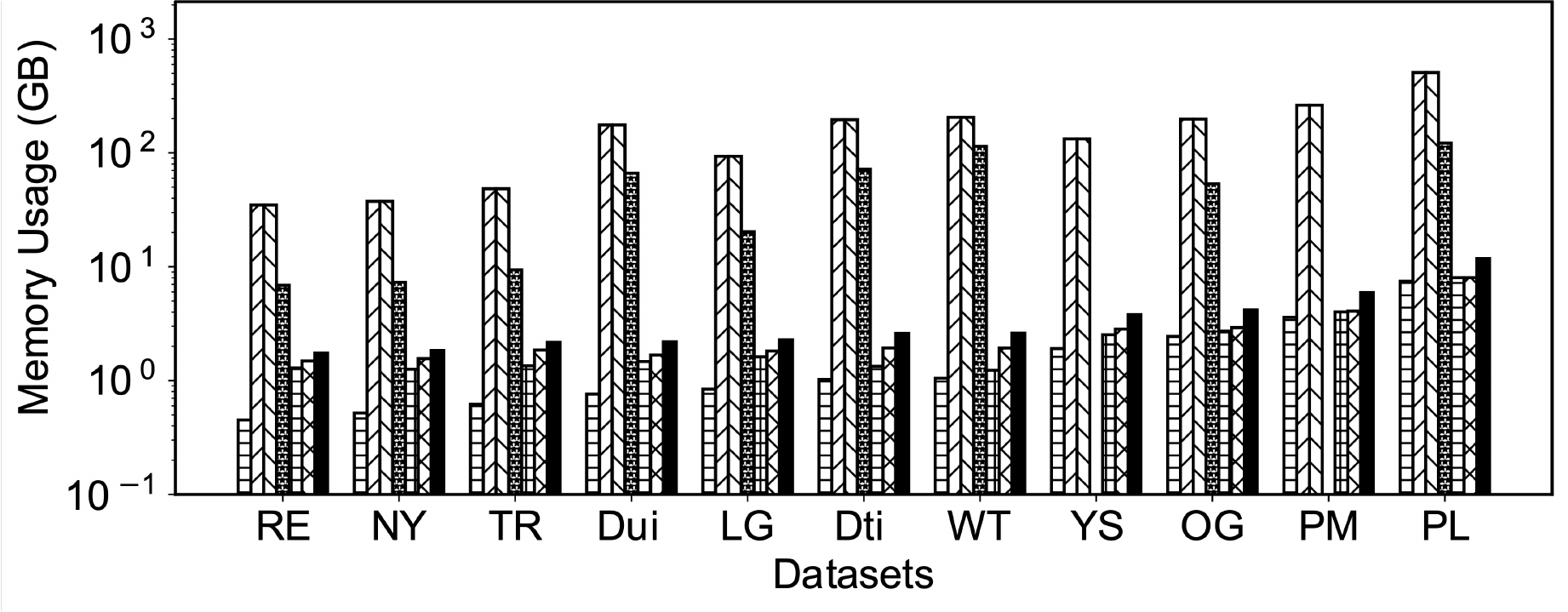}
    \vspace{-2em}
    \caption{Memory usage evaluation on all datasets.}
    \label{fig:mem}
    \vspace{-1.5em}
\end{figure}

\noindent
\textbf{Memory usage during preprocessing.}
Figure~\ref{fig:mem} reports the peak runtime memory consumption during preprocessing.
For \boge, \parindex, and \fang, the reported memory usage corresponds to the peak memory observed during 32-thread index construction.
Although \parindex reduces the memory consumption of parallel index construction compared with \boge, all index-based methods still require substantial intermediate memory.
In particular, the final indexes of \boge and \fang are linear in the number of edges, i.e., $O(m)$, but their construction processes consume much higher peak memory.
As shown in Figure~\ref{fig:mem}, the CPU-based index methods \boge, \parindex, and \fang incur much higher memory overhead than our GPU-based \adv.
For example, on the largest dataset \texttt{PL}, \adv peaks at only 11.7 GB, whereas \boge and \fang each require over 500 GB, and \parindex consumes more than 120 GB.
In addition, since \fang follows the same index construction mechanism as \boge, their memory profiles are nearly identical.
This substantial gap in memory consumption renders index-based methods impractical for large-scale bipartite graphs.
In contrast, \adv keeps its peak memory usage below 15 GB on all datasets, making it suitable for deployment on memory-limited GPUs.
This is because \adv avoids constructing global index structures and only maintains lightweight auxiliary information for online computation.

\begin{figure}[t]
\centering
\includegraphics[width=0.8\linewidth]{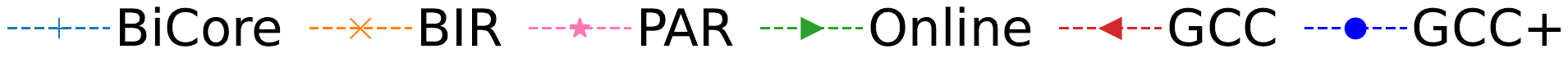}

\subfigure[\texttt{OG}]{
    \includegraphics[width=0.32\linewidth]{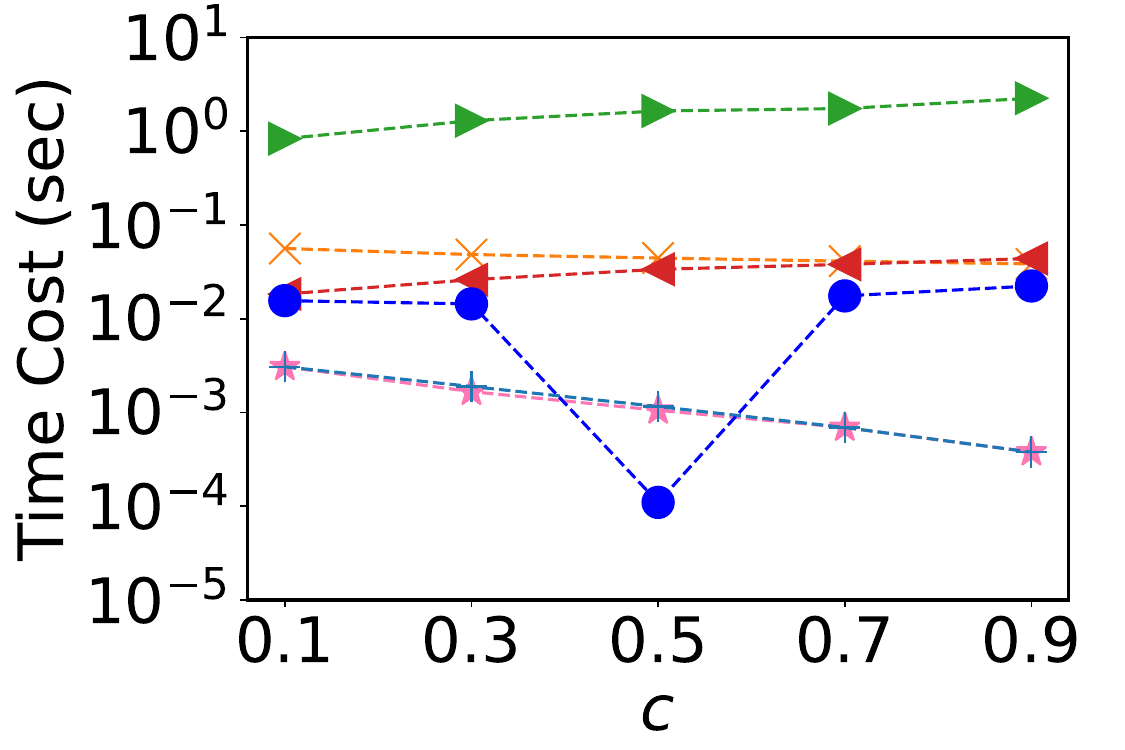}
    \label{fig:p111}
}
\hspace{-14pt}
\subfigure[\texttt{PM}]{
    \includegraphics[width=0.32\linewidth]{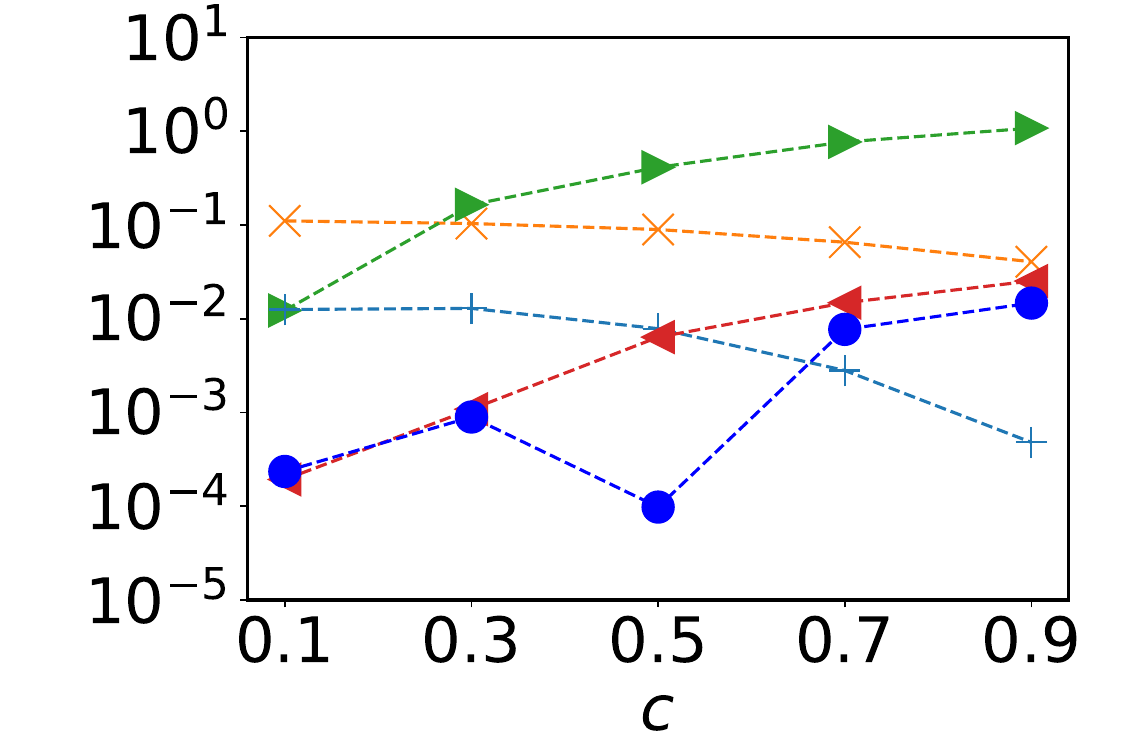}
    \label{fig:p222}
}
\hspace{-14pt}
\subfigure[\texttt{PL}]{
    \includegraphics[width=0.32\linewidth]{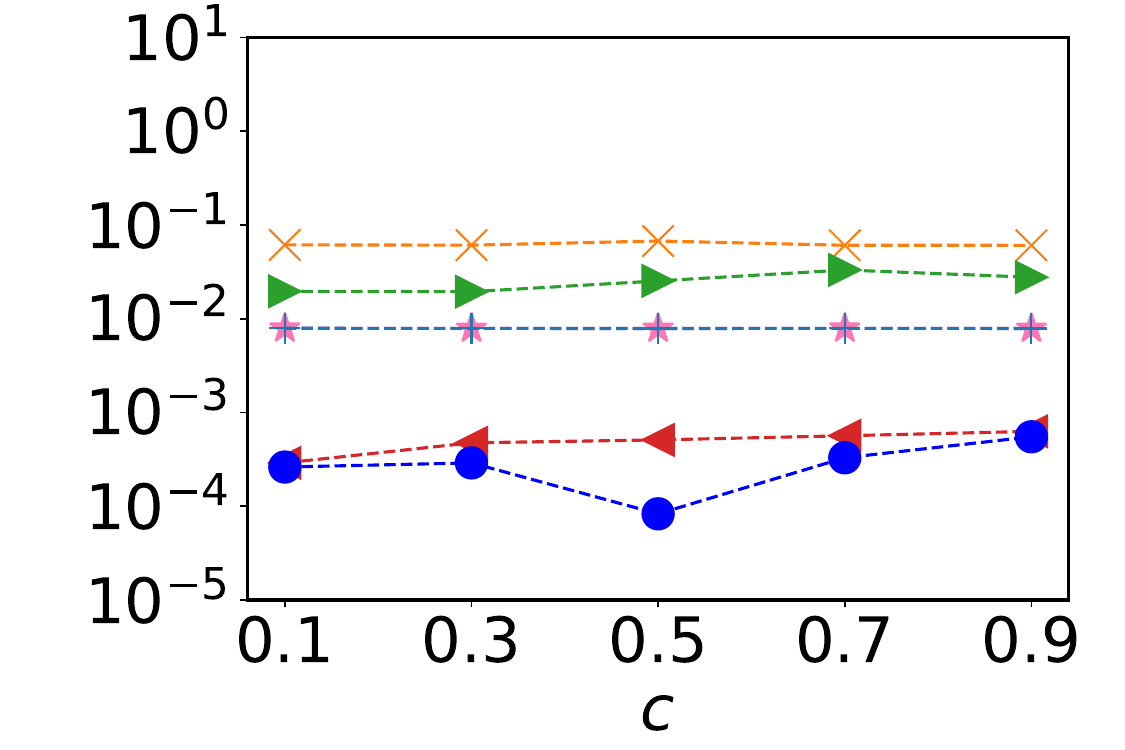}
    \label{fig:p444}
}

\vspace{-3mm}
\caption{\color{black}Performance on retrieving \abcores varying $\alpha$.}
\label{fig:varying_a}
\vspace{-3mm}
\end{figure}

\begin{figure}[t]
\centering
\includegraphics[width=0.8\linewidth]{expfig_new/legend.pdf}

\subfigure[\texttt{OG}]{
    \includegraphics[width=0.32\linewidth]{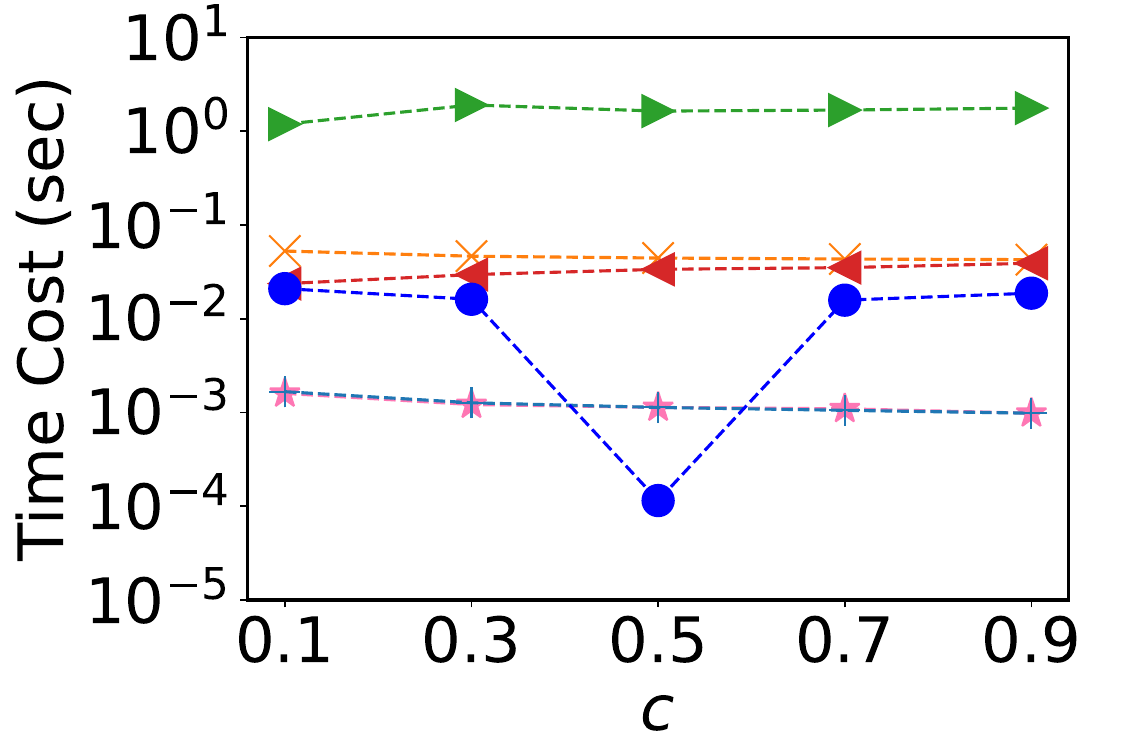}
    \label{fig:p11}
}
\hspace{-14pt}
\subfigure[\texttt{PM}]{
    \includegraphics[width=0.32\linewidth]{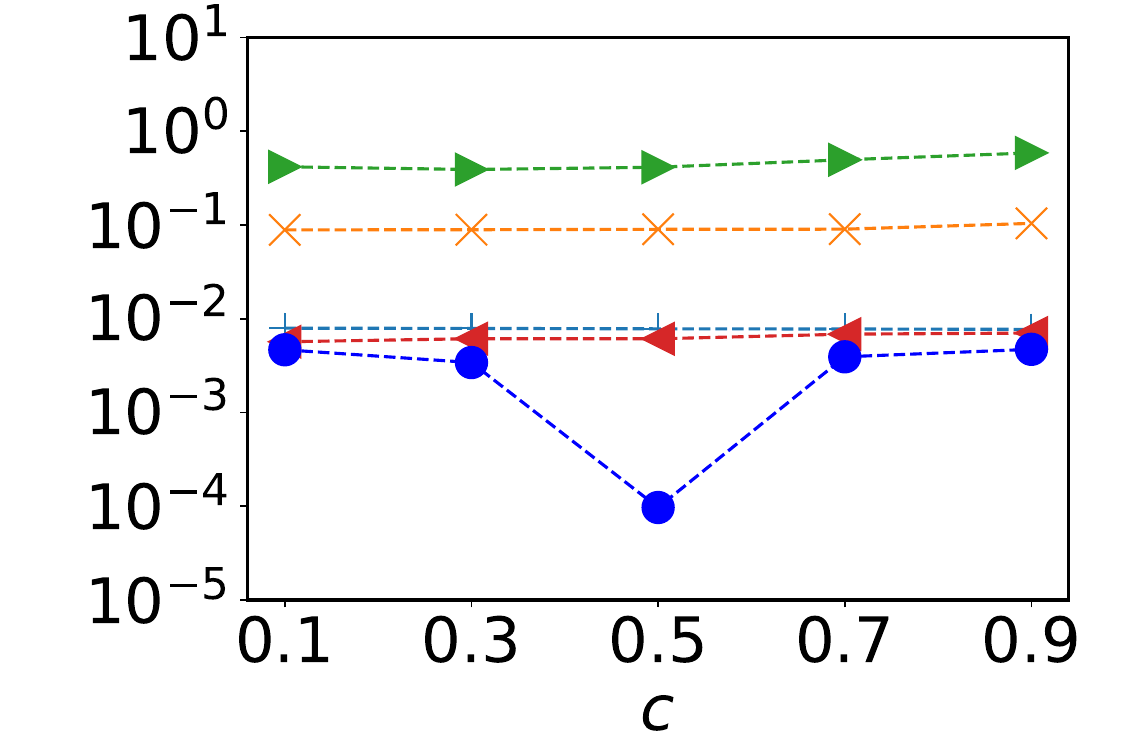}
    \label{fig:p22}
}
\hspace{-14pt}
\subfigure[\texttt{PL}]{
    \includegraphics[width=0.32\linewidth]{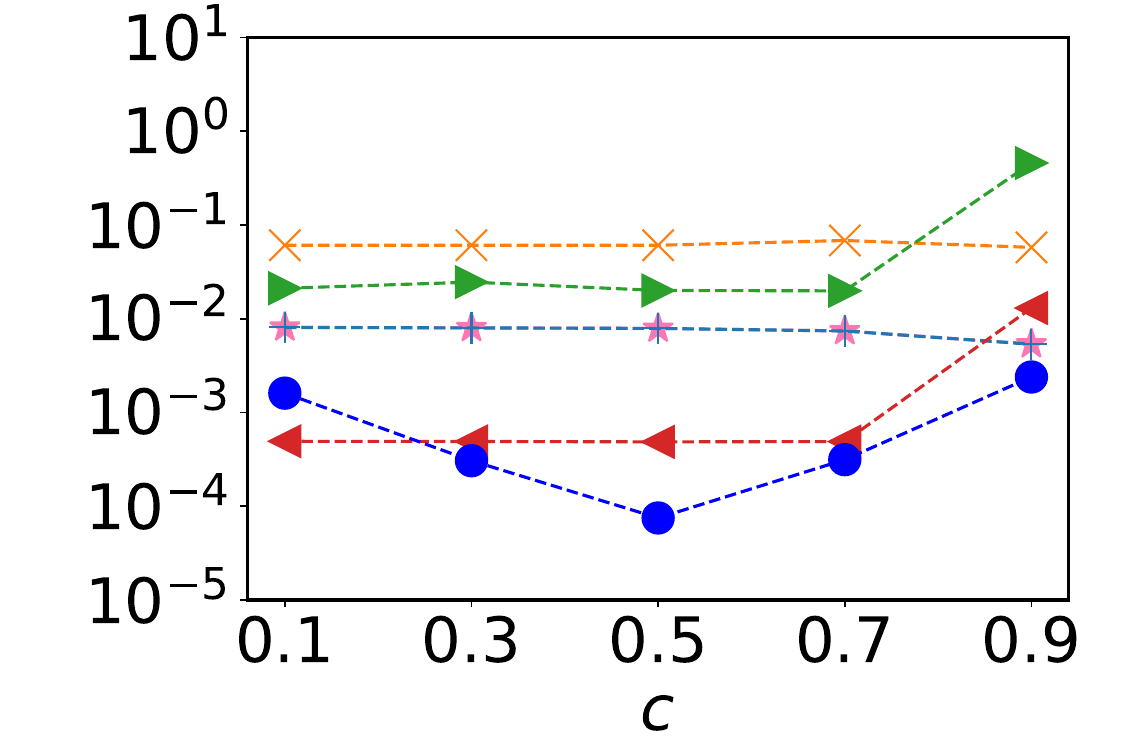}
    \label{fig:p44}
}

\vspace{-3mm}
\caption{\color{black}Performance on retrieving \abcores varying $\beta$.}
\label{fig:varying_b}
\vspace{-1.5em}
\end{figure}

\noindent
\textbf{Query sensitivity to $\alpha$ and $\beta$.}
We further evaluate the query performance under different parameter settings by fixing one parameter and varying the other.
In Figure~\ref{fig:varying_a}, we fix $\beta = 0.5 \times \delta$ and vary $\alpha = c \times \delta$, where $c$ ranges from 0.1 to 0.9.
In Figure~\ref{fig:varying_b}, we fix $\alpha = 0.5 \times \delta$ and vary $\beta = c \times \delta$.
On large datasets, \adv outperforms the index-based methods \boge and \parindex under several parameter settings.
Since \parindex adopts the same index structure as \boge, their query performance is close across different parameter settings.
For example, on the largest dataset \texttt{PL}, when $\alpha = 0.5 \times \delta$ and $\beta = 0.3 \times \delta$, \adv achieves 26$\times$ and 198$\times$ speedups over \boge and \fang, respectively, as shown in Figure~\ref{fig:p44}.
Notably, when $\alpha = \beta = 0.5 \times \delta$, \adv achieves the best performance among all methods because it only requires a simple parallel scan.
These results show that the lightweight preprocessing strategy enables \adv to efficiently handle large-scale bipartite graphs across different $(\alpha,\beta)$ settings.

\subsection{On-the-fly \abcore Query Performance}

\begin{figure}[t]
\centering
\includegraphics[trim=0 0 0 0,width=0.48\textwidth]{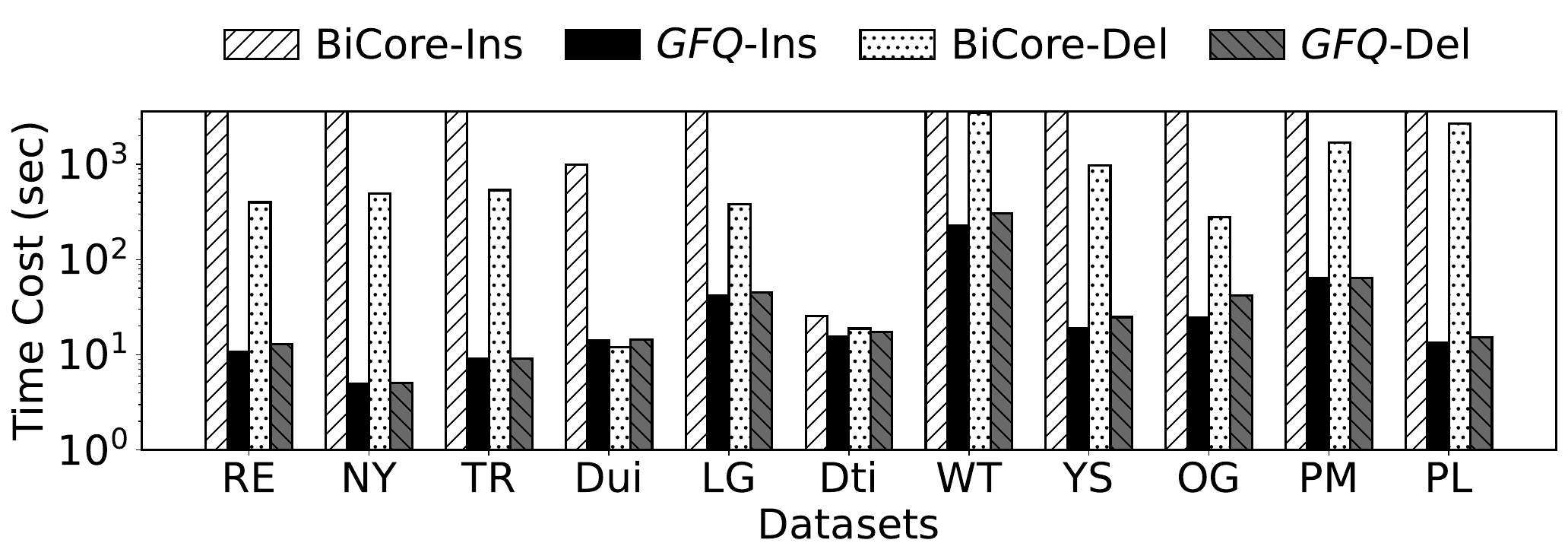}
\vspace*{-2mm}
\caption{Performance on on-the-fly query processing.}
\label{fig:qfly}
\vspace{-1.5em}
\end{figure}

\noindent
\textbf{Query efficiency.}
To evaluate the efficiency of on-the-fly queries under dynamic updates, we randomly select 1,000 edges for insertion and deletion on each dataset.
For each updated edge $e=(u,v)$, the values of $\alpha$ and $\beta$ are randomly chosen such that both $u$ and $v$ belong to the corresponding \abcore.
We set a timeout threshold of 3,600 seconds for the total running time of 1,000 edge insertions or deletions.

As shown in Figure~\ref{fig:qfly}, \qf consistently achieves fast response times for both edge insertions and deletions across all datasets.
In contrast, \boge fails to complete edge insertion within the time limit on most datasets.
For insertion queries, \boge only finishes on \texttt{Dti} and \texttt{Dui}, requiring 25.6 and 994.1 seconds, respectively, while timing out on all other datasets.
By contrast, \qf completes all insertion queries, with running times ranging from 4.94 seconds on \texttt{NY} to 227.9 seconds on \texttt{WT}.
For deletion queries, \qf also consistently achieves lower latency than \boge.
For example, on \texttt{WT}, \qf completes in 305.9 seconds, whereas \boge takes 3,494.5 seconds.
The performance gap mainly comes from the high cost of index maintenance.
In the worst case, the maintenance procedure of \boge has the same $O(\delta \cdot m)$ complexity as rebuilding the entire index~\cite{liu2020efficient}, which is expensive for large and frequently updated graphs.
In contrast, \qf confines the computation to the affected connected components, avoiding global recomputation and enabling efficient on-the-fly query processing under dynamic workloads.

\begin{figure}[t]
    \centering
    \small
    \includegraphics[width=0.48\linewidth]{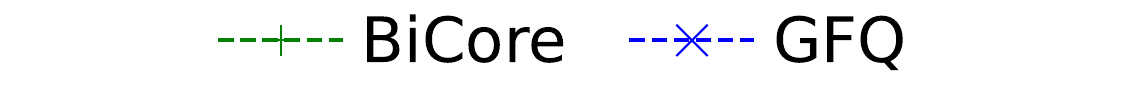}
    \vspace{-1mm}
    
    \subfigure[\texttt{WT} (insertion)]{
        \includegraphics[width=0.34\linewidth]{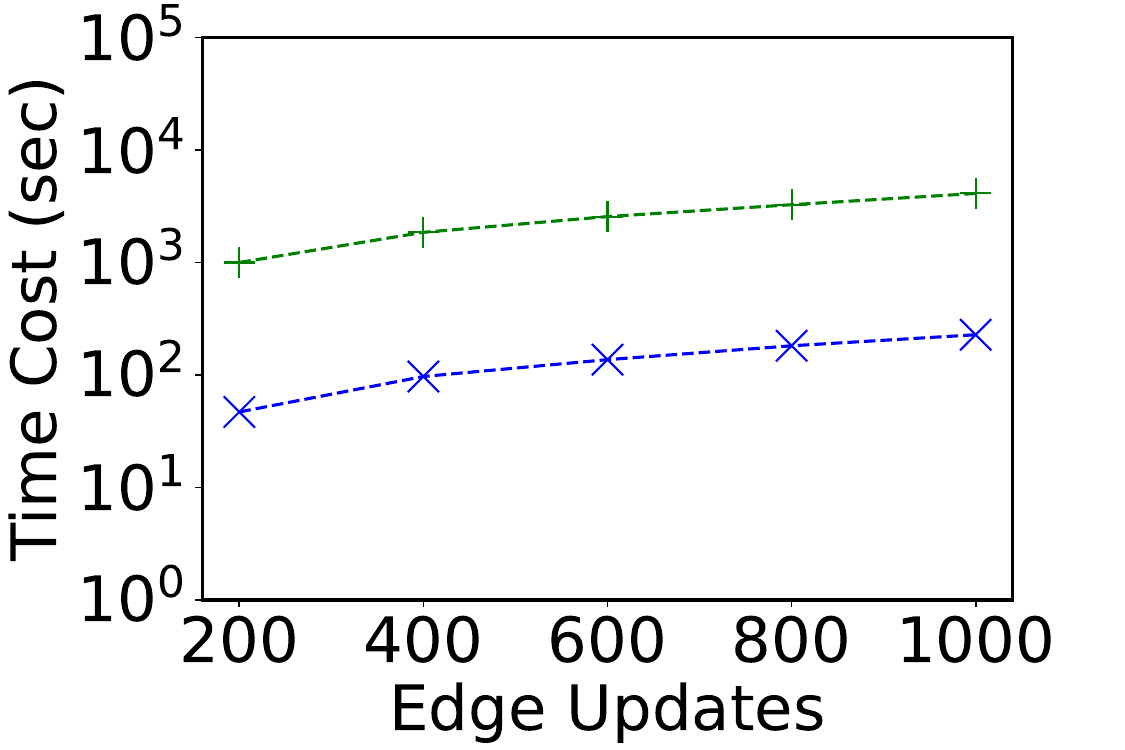}
        \label{fig:v1}
    }
    \hspace{-18pt}
    \subfigure[\texttt{YS} (insertion)]{
        \includegraphics[width=0.34\linewidth]{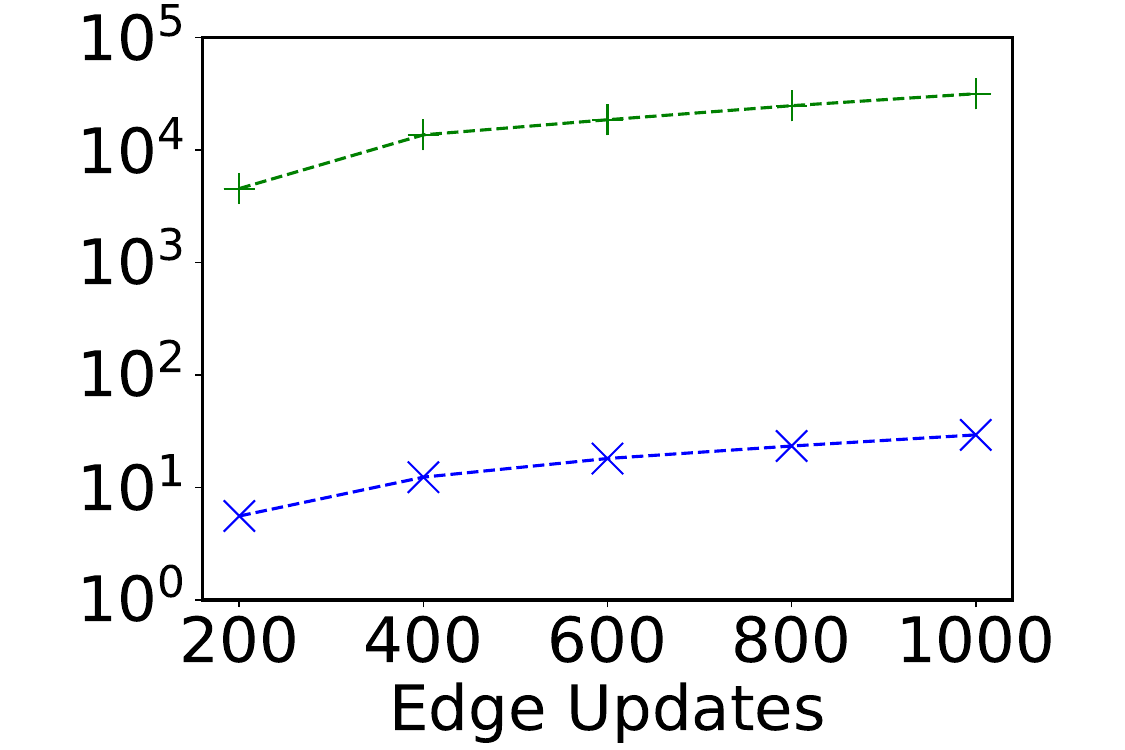}
        \label{fig:v3}
    }
    \hspace{-18pt}
    \subfigure[\texttt{LG} (insertion)]{
        \includegraphics[width=0.34\linewidth]{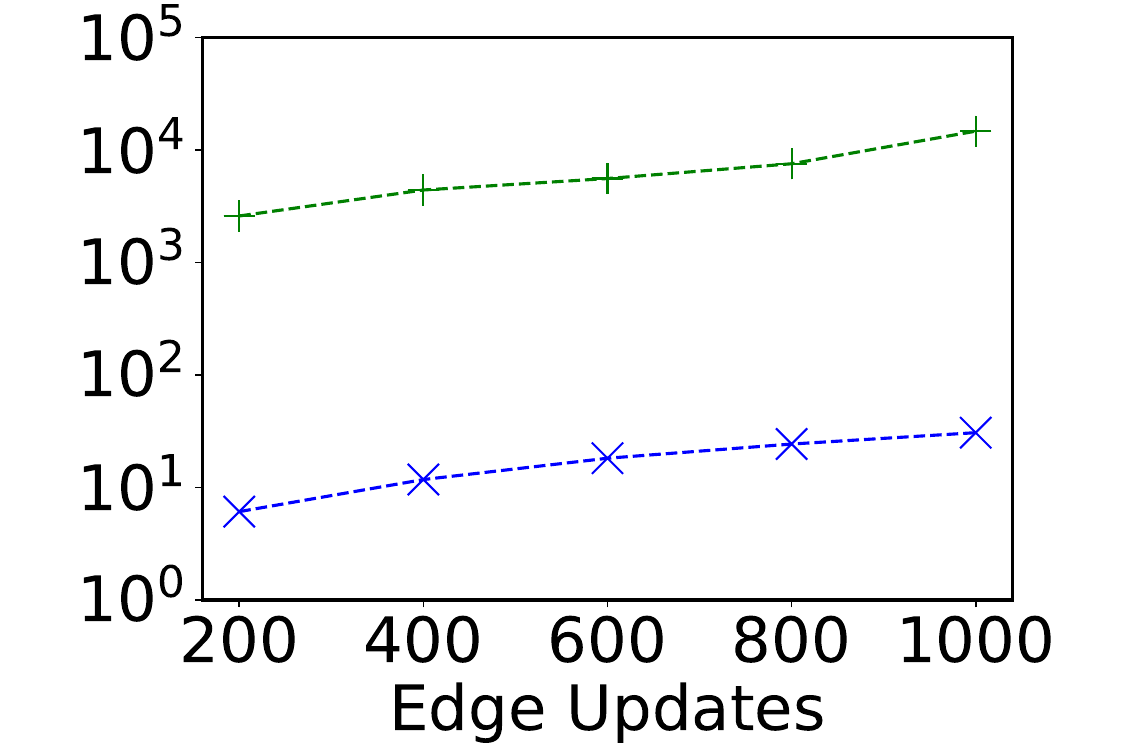}
        \label{fig:v5}
    }

    \vspace{-2mm}

    \subfigure[\texttt{WT} (deletion)]{
        \includegraphics[width=0.33\linewidth]{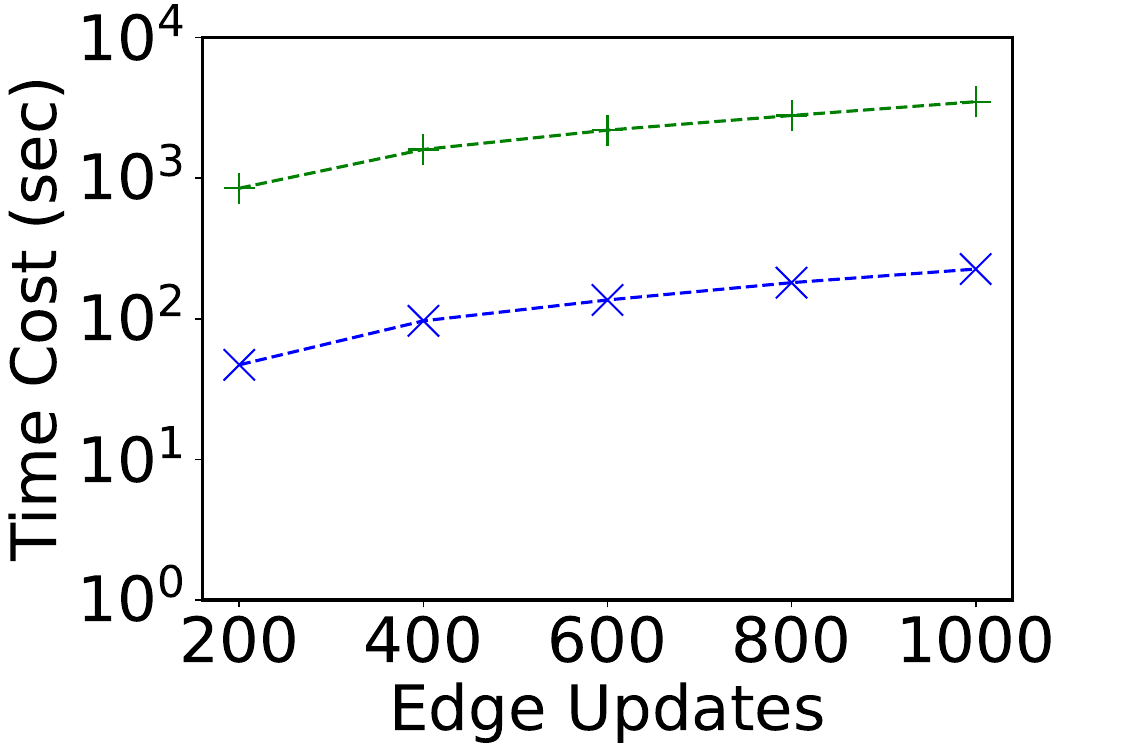}
        \label{fig:v2}
    }
    \hspace{-14pt}
    \subfigure[\texttt{YS} (deletion)]{
        \includegraphics[width=0.33\linewidth]{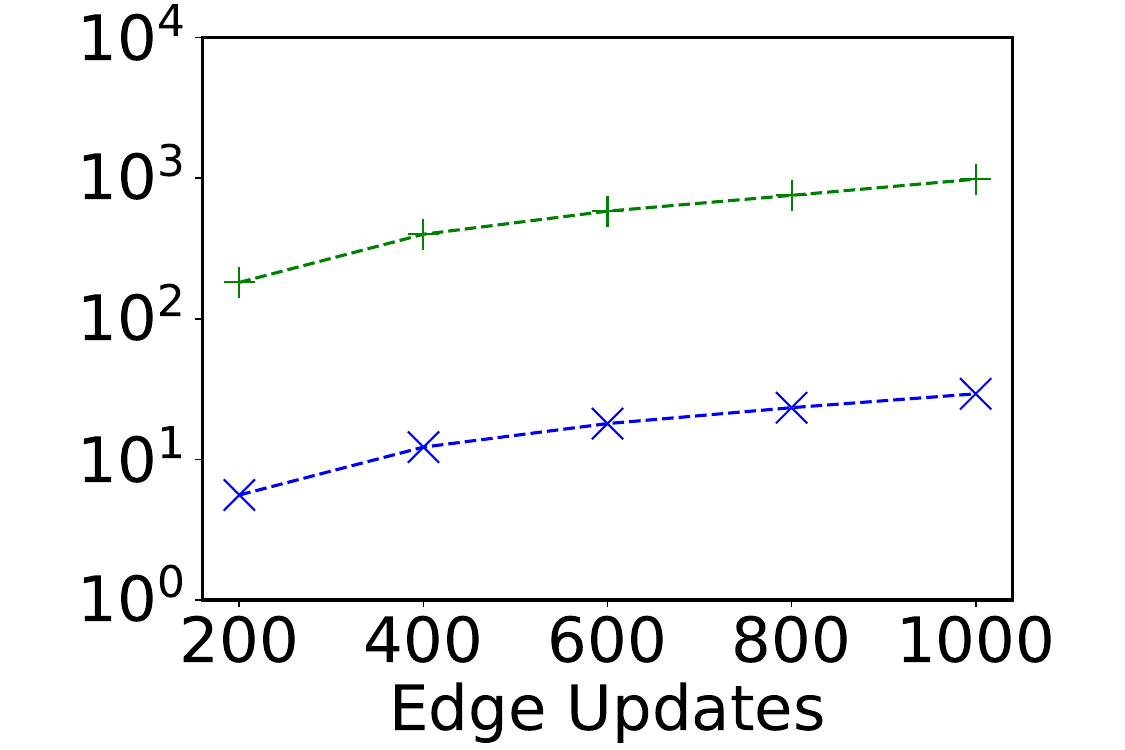}
        \label{fig:v4}
    }
    \hspace{-14pt}
    \subfigure[\texttt{LG} (deletion)]{
        \includegraphics[width=0.33\linewidth]{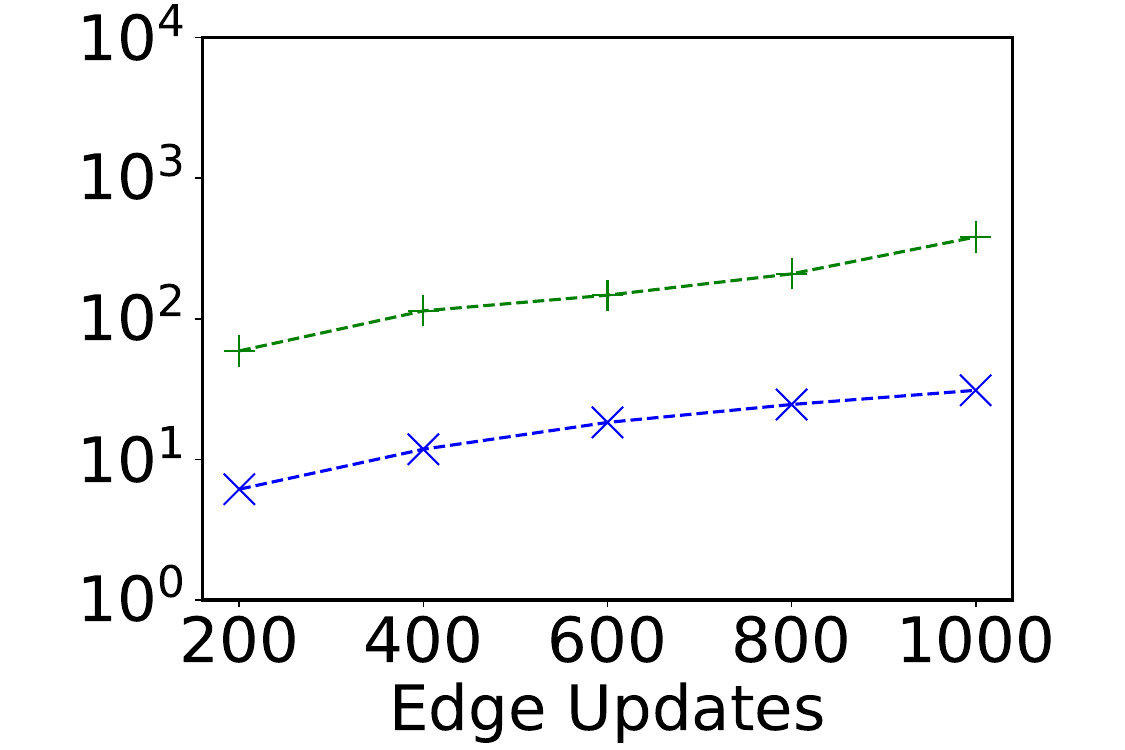}
        \label{fig:v6}
    }
    \caption{The scalability of \qf.}
    \label{fig:scal}
    \vspace{-2em}
\end{figure}

\noindent
\textbf{Scalability performance.}
Figure~\ref{fig:scal} evaluates the scalability of \qf by varying the number of dynamic updates from 200 to 1,000 on \texttt{WT}, \texttt{YS}, and \texttt{LG}.
For each updated edge $e=(u,v)$, the values of $\alpha$ and $\beta$ are randomly chosen such that both $u$ and $v$ belong to the corresponding \abcore.
As shown in Figure~\ref{fig:scal}, \qf consistently outperforms the CPU-based method \boge for both edge insertions and deletions across all tested datasets.
On \texttt{WT}, Figures~\ref{fig:v1} and~\ref{fig:v2} show that the running time of \qf increases smoothly as the number of updates grows; for 1,000 insertions, \qf takes 227.9 seconds, whereas \boge takes 4109.37 seconds.
On \texttt{YS}, the performance gap is even larger: for 1,000 insertions, \qf takes only 29.35 seconds, while \boge requires 31640.3 seconds.
For deletions, \qf keeps the running time below 30 seconds on \texttt{YS} across all update sizes, whereas \boge takes close to 1,000 seconds for 1,000 deletions.
The advantage of \qf comes from its connectivity-aware design, which confines computation to affected connected components and avoids costly global index maintenance.
Overall, these results demonstrate that \qf scales efficiently with the number of dynamic updates and is suitable for on-the-fly query processing on large bipartite graphs.

\section{Conclusion}
\label{sct:conclusion}

In this paper, we propose GPU-based algorithms to efficiently support \abcore computation and on-the-fly query processing on large-scale bipartite graphs. 
We first introduce \baseline, a warp-centric peeling algorithm that eliminates the need for costly index construction. 
To further reduce unnecessary computation, we propose \adv, which performs a lightweight preprocessing based on the core numbers of vertices to prune vertices that are guaranteed to satisfy or violate the degree constraints. 
For dynamic scenarios, we design \qf, a connectivity-aware on-the-fly query algorithm that confines updates and computation to local components. 
Extensive experiments on 11 datasets demonstrate that our algorithms significantly outperform existing solutions.

\bibliographystyle{IEEEtran}
\bibliography{bib/paper}

\end{document}